\documentclass[12pt]{article}

\usepackage{amsmath}
\usepackage{amssymb}
\usepackage{amsfonts}
\usepackage{latexsym}
\usepackage{bm}
\usepackage{color}
\usepackage{theorem}

\catcode `\@=11 \@addtoreset{equation}{section}

\catcode `\@=12

\newcommand{\be}{\begin{equation}}
\newcommand{\en}{\end{equation}}
\newcommand{\bea}{\begin{eqnarray}}
\newcommand{\ena}{\end{eqnarray}}
\newcommand{\beano}{\begin{eqnarray*}}
\newcommand{\enano}{\end{eqnarray*}}
\newcommand{\bee}{\begin{enumerate}}
\newcommand{\ene}{\end{enumerate}}

\newcommand{\mb}{\mathbb}

\newcommand{\R}{\mathbb{R}}
\newcommand{\mc}{\mathcal}

\newcommand{\D}{{\mc D}}

\newcommand{\E}{{\cal E}}
\newcommand{\F}{{\cal F}}

\newcommand{\Lc}{{\cal L}}
\newcommand{\C}{{\mathbb C}}
\newcommand{\1}{1 \!\! 1}

\newcommand{\Hil}{\mc H}

\newtheorem{thm}{Theorem}[section]
\newtheorem{cor}[thm]{Corollary}

\newtheorem{prop}[thm]{Proposition}
\newtheorem{defn}[thm]{Definition}

\newenvironment{proof}{\noindent {\bf Proof --}}{\hfill$\square$ \vspace{3mm}\endtrivlist}

\newcommand{\ip}[2]{\left\langle {#1},{#2} \right\rangle}

\catcode `\@=11 \@addtoreset{equation}{section}
\catcode `\@=12

\begin{document}

\thispagestyle{empty}

\vspace*{2cm}

\begin{center}
{{\Large \bf Generalized Riesz systems\\ and quasi bases in Hilbert space}}\\[10mm]


{\large F. Bagarello} \footnote[1]{ Dipartimento di Energia, Ingegneria dell'Informazione e Modelli Matematici,
Facolt\`a di Ingegneria, Universit\`a di Palermo, I-90128  Palermo, and INFN, Sezione di Napoli, ITALY\\
e-mail: fabio.bagarello@unipa.it\,\,\,\, Home page: www1.unipa.it/fabio.bagarello}
\vspace{3mm}\\

{\large H. Inoue} \footnote[2]{Center for advancing Pharmaceutical Education, Daiichi University of Pharmacy, 22-1 Tamagawa-cho, Minami-ku, Fukuoka 815-8511, Japan\\
e-mail: h-inoue@daiichi-cps.ac.jp}
\vspace{3mm}\\

{\large C. Trapani}\footnote[3]{Dipartimento di Matematica e Informatica, Universit\`a di Palermo, I-90123 Palermo, Italy\\e-mail: camillo.trapani@unipa.it}
\vspace{3mm}\\

\end{center}

\vspace*{2cm}

\begin{abstract}
\noindent The purpose of this article is twofold. First of all, the notion of $(\D, \E)$-quasi basis is introduced for a pair $(\D, \E)$ of dense subspaces of Hilbert spaces. This consists of two biorthogonal sequences $\{ \varphi_n \}$ and $\{ \psi_n \}$ such that $\sum_{n=0}^\infty \ip{x}{\varphi_n}\ip{\psi_n}{y}=\ip{x}{y}$ for all $x \in \D$ and $y \in \E$. Secondly, it is shown that if biorthogonal sequences $\{ \varphi_n \}$ and $\{ \psi_n \}$ form a $(\D ,\E)$-quasi basis, then they are generalized Riesz systems. The latter play an interesting role for the construction of non-self-adjoint Hamiltonians and other physically relevant operators.

\end{abstract}


\vfill


\newpage

\section{Introduction}
A sequence $\{ \varphi_n \}$ in a Hilbert space $\Hil$ is called a generalized Riesz system if there exist an orthonormal basis (from now on, ONB) $\F_e = \{ e_n \}$ in $\Hil$ and a densely defined closed operator $T$ in $\Hil$ with densely defined inverse such that $\F_e \subset D(T) \cap D((T^{-1})^\ast)$ and $Te_n= \varphi_n$, $n=0,1, \cdots$. In this case $(\F_e , T)$ is called a constructing pair for $\{ \varphi_n \}$, \cite{bit_jmp2018, atsushi,hiro2}. Then, if we put $\psi_n := (T^{-1})^\ast e_n$, $n=0,1, \cdots$, $\F_\varphi :=\{ \varphi_n \}$ and $\F_\psi := \{ \psi_n \}$ are biorthogonal sequences in $\Hil$, that is, $\ip{\varphi_n}{\psi_m}=\delta_{nm}$, $n,m=0,1, \cdots$.

The notion of generalized Riesz system is useful to investigate non-self-adjoint Hamiltonians constructed from $\F_\varphi$ and $\F_\psi$. More precisely, let $\F_\varphi$ be a generalized Riesz system with a constructing pair $(\F_e ,T)$ and define $\psi_n$ as above. Then we consider the operators

$$H_\varphi^{\bm{\alpha}}:= TH_{\bm{e}}^{\bm{\alpha}}T^{-1}, \quad A_\varphi^{\bm{\alpha}} := TA_{\bm{e}}^{\bm{\alpha}}T^{-1}\; \mbox{ and }B_\varphi^{\bm{\alpha}} := TB_{\bm{e}}^{\bm{\alpha}}T^{-1},$$
together with
$$H_\psi^{\bm{\alpha}} := (T^\ast)^{-1} H_{\bm{e}}^{\bm{\alpha}}T^\ast, \quad A_\psi^{\bm{\alpha}} := (T^\ast)^{-1} A_{\bm{e}}^{\bm{\alpha}} T^\ast\; \mbox{ and }B_\psi^{\bm{\alpha}} := (T^{-1})^\ast B_{\bm{e}}^{\bm{\alpha}} T^\ast,$$
  where $\bm{\alpha}= \{ \alpha_n \} \subset \C$. Here
$$H_{\bm{e}}^{\bm{\alpha}}:= \sum_{n=0}^\infty \alpha_n e_n \otimes \bar{e}_n, \quad A_{\bm{e}}^{\bm{\alpha}}:= \sum_{n=0}^\infty \alpha_{n+1}e_{n} \otimes \bar{e}_{n+1}, \quad B_{\bm{e}}^{\bm{\alpha}} := \sum_{n=0}^\infty \alpha_{n+1} e_{n+1} \otimes \bar{e}_n$$ are a self-adjoint Hamiltonian, the lowering operator and the raising operator for $\{ e_n \}$, respectively (if, $x,y,z\in \Hil$, $(y\otimes \bar{z})x:=\ip{x}{z}\!y$ ).

Since $H_\varphi^{\bm{\alpha}} \varphi_n =\alpha_n \varphi_n$, $A_\varphi^{\bm{\alpha}} \varphi_n =\alpha_n \varphi_{n-1}$ $(0 \; {\rm if}\; n=0)$ and $B_\varphi^{\bm{\alpha}}\varphi_n =\alpha_{n+1}\varphi_{n+1}$, $n=0,1, \cdots$, it seems natural to call the operators $H_\varphi^{\bm{\alpha}}$, $A_\varphi^{\bm{\alpha}}$ and $B_\varphi^{\bm{\alpha}}$ the non-self adjoint Hamiltonian, and the generalized lowering and raising operators for $\{ \varphi_n \}$, respectively. Similarly, since $H_\psi^{\bm{\alpha}}\psi_n =\alpha_n \psi_n$, $A_\psi^{\bm{\alpha}}\psi_n= \alpha_n\psi_{n-1}$ $(0 \; {\rm if} \;n=0)$ and $B_\psi^{\bm{\alpha}}\psi_n =\alpha_{n+1}\psi_{n+1}$, the operators $H_\psi^{\bm{\alpha}}$, $A_\psi^{\bm{\alpha}}$, $B_\psi^{\bm{\alpha}}$ are called the non-self adjoint Hamiltonian, generalized lowering operator and raising operator for $\{ \psi_n \}$ respectively.\\
Then, it is interesting to understand under what conditions biorthogonal sequences $\F_\varphi$ and $\F_\psi$ are generalized Riesz system, which is what we will discuss in this paper.

Studies on this subject have been undertaken in \cite{atsushi, hiro_taka, hiro1,hiro2}. Here we want to explore this question in a more general framework.

Let $D_\varphi$ and $D_\psi$ be the linear spans of the biorthogonal sequences $\F_\varphi$ and $\F_\psi$, respectively, and define the subspaces $D(\varphi)$ and $D(\psi)$ in $\Hil$ by
\begin{eqnarray}
D(\varphi)
&=& \{ x\in \Hil ; \sum_{n=0}^\infty |\ip{x}{\varphi_n}|^2 < \infty \}, \nonumber \\
D(\psi)
&=& \{ x \in \Hil ; \sum_{n=0}^\infty |\ip{x}{\psi_n}|^2 < \infty \}. \nonumber
\end{eqnarray}
Clearly, $D_\psi \subset D(\varphi)$ and $D_\varphi \subset D(\psi)$. In \cite{hiro1}, one of us has shown that if both $D_{\varphi}$ and $D_{\psi}$ are dense in $\Hil$ (this case is called {\it regular}), then $\F_{\varphi}$ and $\F_{\psi}$ are generalized Riesz systems. After that, in \cite{hiro2}, it was proved that, if either $D_{\varphi}$ and $D(\varphi)$, or $D_{\psi}$ and $D(\psi)$, are dense in $\Hil$ (the case is called {\it semiregular}), again $\F_{\varphi}$ and $\F_{\psi}$ are generalized Riesz systems. Hence we will consider under what conditions $\F_\varphi$ and $\F_\psi$ are generalized Riesz systems when none of the above conditions is satisfied. In \cite{bit_jmp2018}, we have proved that this holds under the assumptions that $\F_\varphi $ and $\F_\psi $ are biorthogonal and, at the same time, $\D$-quasi bases, in the sense that
\begin{eqnarray}
\sum_{n=0}^\infty \ip{x}{\varphi_n}\ip{\psi_n}{y}=\ip{x}{y}, \quad \forall x,y \in \D,\nonumber
\end{eqnarray}
 where $\D$ is a dense subspace in $\Hil$ such that $\F_\varphi \cup \F_\psi \subset \D \subset D(\varphi) \cap D(\psi)$, with some additional assumptions.
 In this paper we shall show that this result holds in a more general case. In Section 3 we define the notion of $(\D,\E)$-quasi bases which is a generalization of $\D$-quasi bases as follows:
\begin{eqnarray}
\sum_{n=0}^\infty \ip{x}{\varphi_n}\ip{\psi_n}{y} =\ip{x}{y}, \quad \forall x\in \D, \,y \in \E\nonumber
\end{eqnarray}
where $\D$ and $\E$ are dense subspaces in $\Hil$ such that $D_\psi \subset \D \subset D(\varphi)$ and $D_\varphi \subset \E \subset D(\psi)$, and we show in Theorem 3.2 that, under this condition, $\F_\varphi$ and $\F_\psi$ are generalized Riesz systems.

In Section 4, we shall investigate non-self adjoint Hamiltonians, generalized lowering and raising operators constructed from $(\D,\E)$-quasi bases. This analysis can be relevant for concrete physical applications, and extends what already deduced, for instance, in \cite{bit2013,hiro1,bb2017}.


\section{Preliminaries}
In this section we review some results on generalized Riesz systems needed in the rest of the paper. By Lemma 3.2, \cite{hiro2}, we have the following\\
\par
{\bf Lemma 2.1.} {\it Let $\{ \varphi_{n} \}$ be a generalized Riesz basis with a constructing pair $(\F_e,T)$. Then, we have the following statements.
\par
\hspace{3mm} (1) $T^{\ast}$ has a densely defined inverse and $(T^{\ast})^{-1}= (T^{-1})^{\ast}$.
\par
\hspace{3mm} (2) Let $\psi_{n} := (T^{-1})^{\ast} e_{n}$, $n=0,1, \cdots$.
Then, $\{ \varphi_{n} \}$ and $\{ \psi_{n} \}$ are biorthogonal and $(T^{-1})^{\ast}$ is a densely defined closed operator in ${\cal H}$ with densely defined inverse $T^{\ast}$. Hence $\{ \psi_{n} \}$ is a generalized Riesz basis with a constructing pair $(\F_e , (T^{-1})^{\ast} )$.
\par
\hspace{3mm} (3) $D(\varphi) \cap D(\psi)$ is dense in ${\cal H}$.\\}
\par
Next, for any ONB $\{ e_n \}$ in $\Hil$ and a sequence $\{ \varphi_n \}$ in $\Hil$, we introduce the operators $T_{\varphi,\bm{e}}^0$, $T_{\varphi,\bm{e}}$ and $T_{\bm{e},\varphi}$ as follows:
\begin{eqnarray}
T_{\varphi,\bm{e}}^0
&:=& \; {\rm the \; linear \; operator \; defined \; by} \;\;\; T_{\varphi,\bm{e}}^0 e_n =\varphi_n,\;\;\; n=0,1, \cdots , \nonumber \\
T_{\varphi,\bm{e}}
&:=& \sum_{n=0}^\infty \varphi_n \otimes \bar{e}_n , \nonumber \\
T_{\bm{e},\varphi}
&:=& \sum_{n=0}^\infty e_n \otimes \bar{\varphi}_n. \nonumber
\end{eqnarray}
Similarly we can introduce, for the set $\{\psi_n\}$ in  Lemma 2.1, the operators$T_{\psi,\bm{e}}^0$, $T_{\psi,\bm{e}}$ and $T_{\bm{e},\psi}$. These operators had  a role  in \cite{hiro2} and will also be relevant here. By Lemmas 2.1, 2.2 in \cite{hiro2} we get the following\\
\par
{\bf Lemma 2.2.} {\it (1) $T_{\varphi,\bm{e}}$ is a densely defined linear operator in $\Hil$ such that
\begin{eqnarray}
T_{\varphi,\bm{e}} \supseteq T_{\varphi,\bm{e}}^0 \;\;\; {\rm and} \;\;\; T_{\varphi,\bm{e}}^0 e_n =T_{\varphi,\bm{e}} e_n =\varphi_n , \;\;\; n=0,1, \cdots . \nonumber
\end{eqnarray}
\par
(2) $D(T_{\bm{e},\varphi})=D(\varphi)$ and $(T_{\varphi,\bm{e}}^0)^\ast =T_{\varphi,\bm{e}}^\ast =T_{\bm{e},\varphi}$.\\
\par
(3) $T_{\varphi,\bm{e}}^0$ is closable if and only if $T_{\varphi,\bm{e}}$ is closable if and only if $D(\varphi)$ is dense in $\Hil$. If this holds, then
\begin{eqnarray}
\bar{T}_{\varphi,\bm{e}}^0 = \bar{T}_{\varphi,\bm{e}} =(T_{\bm{e},\varphi})^\ast.
\end{eqnarray}
}
Furthermore, by Lemmas 2.3 and 2.4 in \cite{hiro2} we have\\
\par
{\bf Lemma 2.3.} {\it Let $\F_\varphi$ and $\F_\psi$ be biorthogonal sequences in $\Hil$. Suppose that $D(\varphi)$ is dense in $\Hil$. Then we have the following
\par (1) $\bar{T}_{\varphi,\bm{e}}$ has an inverse and $\bar{T}_{\varphi,\bm{e}}^{-1} \subseteq T_{\bm{e},\psi} =(T_{\psi,\bm{e}})^\ast$.\\
\par
(2) The following (i), (ii) and (iii) are equivalent:
\par
(i) $D_{\phi} $ is dense in ${\cal H}$.
\par
(ii) $\bar{T}_{\varphi,\bm{e}}$ has a densely defined inverse.
\par
(iii) $T_{\varphi,\bm{e}}^\ast (=T_{\bm{e},\varphi})$ has a densely defined inverse.\\
If this holds, then $T_{\bm{e},\varphi}^{-1}=(\bar{T}_{\varphi,\bm{e}}^{-1})^\ast$.\\
\par
(3) For the operators $T_{\psi,\bm{e}}$ and $T_{\bm{e},\psi}$ the same results as in (1) and (2) hold.\\\\
}
By \cite{hiro2}, Theorem 3.4, we also get\\
\par
{\bf Theorem 2.4.} {\it Let $\F_\varphi$ and $\F_\psi$ be biorthogonal sequences in ${\cal H}$, and let $\F_e$ be an arbitrary ONB in ${\cal H}$. Then the following statements hold:
\par
\hspace{3mm} (1) Suppose that both $D_\varphi$ and $D_\psi$ are dense in $\Hil$. Then $\F_\varphi$ (resp. $\F_\psi$) is a generalized Riesz basis with constructing pairs $(\F_e, \bar{T}_{\phi,\bm{e}})$ and $(\F_e, T_{\bm{e},\psi}^{-1})$ (resp. $(\F_e, \bar{T}_{\psi,\bm{e}})$ and $(\F_e, T_{\bm{e},\phi}^{-1})$), and $\bar{T}_{\phi,\bm{e}}$ (resp. $\bar{T}_{\psi,\bm{e}}$) is the minimum among constructing operators of the generalized Riesz basis $\F_\varphi$ (resp. $\F_\psi$), and $T_{\bm{e},\psi}^{-1}$ (resp. $T_{\bm{e},\phi}^{-1}$) is the maximum among constructing operators of $\F_\varphi$ (resp. $\F_\psi$). Furthermore, any closed operator $T$ (resp. $K$) satisfying $\bar{T}_{\phi,\bm{e}} \subset T \subset T_{\bm{e},\psi}^{-1}$ (resp. $\bar{T}_{\psi,\bm{e}} \subset K \subset T_{\bm{e},\phi}^{-1}$) is a constructing operator for $\F_\varphi$ (resp. $\F_\psi$).
\par
\hspace{3mm} (2) Suppose that $D(\phi)$ and $D_{\phi}$ are dense in ${\cal H}$. Then $\F_\varphi$ (resp. $\F_\psi$) is a generalized Riesz basis with a constructing pair $(\F_e, \bar{T}_{\phi,\bm{e}})$ (resp. $(\F_e , T_{\bm{e}, \phi}^{-1} )$) and the constructing operator $\bar{T}_{\phi ,\bm{e}}$ (resp. $T_{\bm{e},\phi}^{-1}$) is the minimum (resp. the maximum) among constructing operators of $\F_\varphi$ (resp. $\F_\psi$).
\par
\hspace{3mm} (3) Suppose that $D(\psi)$ and $D_{\psi}$ are dense in ${\cal H}$. Then $\F_\psi$ (resp. $\F_\varphi$) is a generalized Riesz basis with a constructing pair $(\F_e, \bar{T}_{\psi,\bm{e}})$ (resp. $( \F_e , T_{\bm{e}, \psi}^{-1} )$) and the constructing operator $\bar{T}_{\psi ,\bm{e}}$ (resp. $T_{\bm{e},\psi}^{-1}$) is the minimum (resp. the maximum) among constructing operators of $\F_\psi$ (resp. $\F_\varphi$).\\\\
}
Theorem 2.4 shows how the problem stated in Introduction (under what conditions biorthogonal sequences $\F_\varphi$ and $\F_\psi$ are generalized Riesz systems) can be solved in the case when either $D_\varphi$ and $D(\psi)$ or $D_\psi$ and $D(\varphi)$ are dense in $\Hil$. But, this problem has not been solved completely  in case that both $D_\varphi$ and $D_\psi$ are not dense in $\Hil$, which is what is interesting for us here. We will see how the operators $T_{\varphi,\bm{e}}$, $T_{\bm{e},\varphi}$, $T_{\psi,\bm{e}}$ and $T_{\bm{e},\psi}$ will be relevant in our analysis, together with the $(\D,\E)$-quasi bases we will define in the next section. This result is a generalization of the one obtained in \cite{bit_jmp2018}.

\section{$(\D,\E)$-quasi bases}\label{sect3}

In this section we extend the notion of $\D$-quasi bases by introducing a second dense subset $\E$ of the Hilbert space $\Hil$, and we relate these new families of vectors to generalized Riesz systems.

 \par
\begin{defn} Let $\F_\varphi$ and $\F_\psi$ be biorthogonal sequences in $\Hil$ and let $\D$ and $\E$ be dense subspaces such that $D_\psi \subseteq \D \subseteq D(\varphi)$ and $D_\varphi \subseteq \E \subseteq D(\psi)$. Then $(\{\varphi_n\},\{\psi_n\})$ is said to be a $(\D,\E)$-quasi basis if
\begin{eqnarray}
\sum_{k=0}^{\infty} \ip{x}{\varphi_k}\ip{\psi_k}{y}=\ip{x}{y} \nonumber
\end{eqnarray}
for all $x\in \D$ and $y\in\E$.
\end{defn}

It is clear that any  $(\D,\D)$-quasi basis is  a $\D$-quasi basis in the sense of \cite{bag1}.

\vspace{2mm}

{
{\bf Example 1:}-- A very simple example of a $(\D,\E)$-quasi basis can be constructed as follows.
Let $\{e_n\}$ be an ONB for $\Hil$. Let ${\alpha_n}$ an unbounded sequence of positive real numbers having $0$ as limit point. To be more concrete, let us take
$$ \alpha_n=\left\{\begin{array}{ll} \frac{1}{n} & \mbox{ if $n$ is even} \\ n &\mbox{ if $n$ is odd.}  \end{array} \right.$$
Let $Tx=\sum_{n=1}^\infty \alpha_n \ip{x}{e_n}e_n$ be defined on the domain
$$D(T)=\left\{x\in \Hil: \sum_{k=0}^\infty (2k+1)^2|(x,e_{2k+1})|^2<\infty\right\}.$$
The operator $T$ is unbounded, selfadjoint, invertible with inverse $T^{-1}$ is defined as $T^{-1}y=\sum_{n=1}^\infty \alpha_n^{-1} \ip{x}{e_n}e_n$ on the domain
$$D(T^{-1})=\left\{y\in \Hil: \sum_{k=1}^\infty (2k)^2|(y,e_{2k})|^2<\infty\right\}.$$
Both $D(T)$ and $D(T^{-1})$ are dense subspaces of $\Hil$ and they are different as one can easily check. Let us set $\varphi_n=Te_n$ and $\psi_n=T^{-1}e_n$, $n\in {\mb N}$. The $\varphi_n=\alpha_ne_n$, while $\psi_n=T^{-1}e_n=\alpha_n^{-1}e_n$. Moreover $D(\varphi)=D(T)$, $D(\psi)=D(T^{-1})$.
Then we have
\begin{eqnarray*}
\sum_{n=0}^{\infty} \ip{x}{\varphi_n}\ip{\psi_n}{y}=\sum_{n=0}^{\infty} \ip{x}{\alpha_ne_n}\ip{\alpha_n^{-1}e_n}{y}=\ip{x}{y}.
\end{eqnarray*}
Thus, $(\F_\varphi, \F_\psi)$ is a $(D(\varphi), D(\psi))$-quasi basis.

\vspace{2mm}

{\bf Example 2:}-- Let $H_0=p^2+x^2$ be (twice) the self-adjoint Hamiltonian  of a one-dimensional harmonic oscillator. We consider $H_0$ to be the closure of the operator acting in the same way on the Schwartz space ${\mc S}({\mb R})$, and $T=\1+p^2$, which is an unbounded self-adjoint operator defined on $D(T)=W^{2,2}({\mb R})$, the Sobolev space of functions having first and second order weak derivative in $L^2({\mb R})$. The operator $T=H_0+\1-x^2$ is unbounded, invertible with bounded inverse $T^{-1}$. The eigensystem of $H_0$ is well known:
$$
H_0e_n(x)=(2n+1)e_n(x), \qquad e_n(x)=\frac{1}{\sqrt{2^nn!\pi^{1/2}}}\,H_n(x)\,e^{-x^2/2}
$$
$n\geq0$, where $H_n(x)$ is the n-th Hermite polynomial.
{ Moreover,
\begin{equation} \label{eqn_Hzero}H_0 f= \sum_{n=0}^\infty (2n+1) (e_n \otimes \bar{e}_n)f= \sum_{n=0}^\infty (2n+1)(f,e_n)e_n, \; \forall f\in {\mc S}({\mb R}).\end{equation}

}

It is easy to see that $e_n(x)\in D(T)$ so that we can define $\varphi_n(x)=(Te_n)(x)$ and $\psi_n(x)=(T^{-1}e_n)(x)$. We get
$$
\varphi_n(x)=(2+2n-x^2)e_n(x),\qquad \psi_n(x)=\frac{1}{2}\int_{\Bbb R}e^{-|x-y|}\,e_n(y)\,dy.
$$
These functions are respectively eigenvectors of $H=TH_0T^{-1}$ and $H^\dagger$, with eigenvalue $2n+1$. Some computations show that, for instance,
$$
H=H_0-2\left(\1+2x\frac{d}{dx}\right)\,G\star.
$$
Here $G(x)$ is the Green function of $T$, $G(x)=\frac{1}{2}e^{-|x|}$, and $(G\star f)(x)=\int_{\Bbb R}G(x-y)f(y)dy$, for all $f(x)\in L^2(\Bbb R)$. Of course we can rewrite $H$ as follows: $H=H_0-2(\1+2ixp)\,G\star$, which is manifestly non self-adjoint.

The sets $\F_\varphi$ and $\F_\psi$ are biorthogonal and form a $(D(T),\Hil)$-quasi basis, since
$$
\sum_{k=0}^{\infty} \ip{f}{\varphi_k}\ip{\psi_k}{g}=\ip{f}{g},
$$
for all $f(x)\in D(T)$ and $g(x)\in L^2(\Bbb R)$.
}

\medskip
Let $\F_\varphi$ and $\F_\psi$ be biorthogonal sequences. Suppose that $\F_\varphi$ is a generalized Riesz system with constructing pair $(\F_{\bm{e}},T)$. We put $\psi^T_n := (T^{-1})^\ast e_n$, $n=0,1, \cdots$. Then $\F_\psi$ and $\F_\psi^T := \{ \psi^T_n \}$ are biorthogonal sequences, but $\F_\psi$ does not necessarily coincide with $\F_\psi^T$. For this reason we will call the constructing pair $(\F_{\bm{e}},T)$ natural for the biorthogonal sequences $\F_\varphi$ and $\F_\psi$ if $\F_\psi =\F_\psi^T$.
If $D_\varphi$ is dense in $\Hil$, then $(\F_e,T)$ is automatically natural for $\F_\varphi$ and $\F_\psi$.

The next theorem, which is the main result of this paper, shows that the notion of $(\D,\E)$-quasi basis is intimately linked to that of generalized Riesz system.

\begin{thm} Let $(\F_\varphi,\F_\psi)$ be a biorthogonal pair and $\D$ and $\E$ be dense subspaces in $\Hil$ such that $D_\psi \subseteq\D\subseteq D(\varphi)$ and $D_\varphi\subseteq \E\subseteq D(\psi)$. Then the following statements are equivalent:
\par
(i) $(\F_\varphi,\F_\psi)$ is a $(\D,\E)$-quasi basis.
\par
(ii) For any ONB $\F_e=\{e_n\}$ in $\Hil$, $\F_\varphi$ is a generalized Riesz system with a natural constructing pair $(\F_e,T)$ satisfying $D(T^{\ast}) \supseteq \D$ and $D(T^{-1}) \supseteq\E$.
\par
(iii) For any ONB $\F_e=\{e_n\}$ in $\Hil$, $\F_\psi$ is a generalized Riesz system with a natural constructing pair $(\F_e,K)$ satisfying $D(K^{\ast}) \supseteq \E$ and $D(K^{-1}) \supseteq \D$.\\\\
If the statement (i) holds, then we can take $(\overline{T_{\bm{e},\psi}\lceil_\E})^{-1}$ and $(\overline{T_{\bm{e},\varphi}\lceil_\D})^{-1}$ as $T$ and $K$ in (ii) and (iii), respectively. If $D_\varphi$ is not dense in $\Hil$, then $T_{\bm{e},\psi}$ does not have an inverse, but $T_{\bm{e},\psi}\lceil_\E$ has an inverse.
\end{thm}

\begin{proof} Take arbitrary $x\in\D$ and $y\in\E$. Since $x\in D(T_{\bm{e},\varphi})=D(\varphi)$ and $y\in D(T_{\bm{e},\psi})=D(\psi)$, we have
\begin{eqnarray}
\ip{x}{y}
&=& \sum_{n=0}^{\infty} \ip{x}{\varphi_n}\ip{\psi_n}{y}
= \sum_{n=0}^{\infty} \ip{x}{T_{\varphi,\bm{e}}e_n}\ip{T_{\psi,\bm{e}}e_n}{y} \nonumber \\
&=& \sum_{n=0}^{\infty} \ip{T_{\bm{e},\varphi}x}{e_n}\ip{e_n}{T_{\bm{e},\psi}y}
= \ip{T_{\bm{e},\varphi}x}{T_{\bm{e},\psi}y}, \nonumber
\end{eqnarray}
which implies that
\begin{eqnarray}
(\overline{T_{\bm{e},\psi}\lceil_\E})^{-1} \subseteq (T_{\bm{e},\varphi}\lceil_\D)^{\ast} \;\;\; {\rm and} \;\;\; (\overline{T_{\bm{e},\varphi}\lceil_\D})^{-1} \subseteq (T_{\bm{e},\psi}\lceil_\E)^{\ast}.
\end{eqnarray}
Now we put $T := (\overline{T_{\bm{e},\psi}\lceil_\E})^{-1}$. Since $D(T)= \overline{T_{\bm{e},\psi}\lceil_\E}D(\overline{T_{\bm{e},\psi}\lceil_\E}) \supseteq \overline{T_{\bm{e},\psi}\lceil_\E} \E \supseteq  \overline{T_{\bm{e},\psi}\lceil_\E} D_\varphi=D_{\bm{e}}$ and $D((T^{-1})^\ast)= D(  (\overline{T_{\bm{e},\psi}\lceil_\E})^{\ast}) \supseteq D( (\overline{T_{\bm{e},\varphi}\lceil_\D})^{-1} )=\overline{T_{\bm{e},\varphi}\lceil_\D}D(\overline{T_{\bm{e},\varphi}\lceil_\D}) \supseteq \overline{T_{\bm{e},\varphi}\lceil_\D}D_\psi=D_{\bm{e}}$, it follows that $T$ is a densely defined closed operator in $\Hil$ with densely defined inverse such that $\bm{e} \subseteq D(T) \cap D((T^{-1})^{\ast})$. Furthermore, we have
\begin{eqnarray}
Te_n
&=& (\overline{T_{\bm{e},\psi}\lceil_\E})^{-1}\overline{T_{\bm{e},\psi}\lceil_\E} \varphi_n =\varphi_n, \nonumber \\
(T^{-1})^{\ast}e_n
&=& (\overline{T_{\bm{e},\psi}\lceil_\E})^{\ast}e_n= T_{\psi,\bm{e}}e_n=\psi_n, \;\;\; n=0,1, \cdots . \nonumber
\end{eqnarray}
Thus, $\F_\varphi$ is a generalized Riesz system with a natural constructing pair $(\F_e,T)$. Furthermore, we have $D(T^{-1})=D(\overline{T_{\bm{e},\psi}\lceil_\E}) \supseteq \E$ and by (3.1) $D(T^{\ast}) \supseteq D(\overline{T_{\bm{e},\varphi}\lceil_\D}) \supseteq \D$. Thus (i) $\Rightarrow$ (ii).\\
 In a similar way, setting $K= (\overline{T_{\bm{e},\varphi}\lceil_\D})^{-1}$, we can show that $\F_\psi$ is a generalized Riesz system for a natural constructing pair $(\F_e,K)$ satisfying $D(K^{\ast}) \supseteq \E$ and $D(K^{-1}) \supseteq \D$. Thus (i) implies (iii).\\
(ii) $\Rightarrow$ (i) Take arbitrary $x\in\D$ and $y\in\E$. Since
\begin{eqnarray}
\sum_{k=0}^{\infty} \ip{x}{\varphi_k}\ip{\psi_k}{y}
&=& \sum_{k=0}^{\infty} \ip{x}{Te_n}\ip{(T^{-1})^{\ast} e_n}{y} \nonumber \\
&=& \sum_{k=0}^{\infty}\ip{T^{\ast}x}{e_n}\ip{e_n}{T^{-1}y}
= \ip{T^{\ast}x}{T^{-1}y}
= \ip{x}{y}, \nonumber
\end{eqnarray}
it follows that $(\F_\varphi,\F_\psi)$ is a $(\D,\E)$-quasi basis. Similarly we can show (iii) $\Rightarrow$ (i). This completes the proof.
\end{proof}

For $\D$-quasi basis, we have the following \\
\par
\begin{cor} Let $\F_\varphi$ and $\F_\psi$ be biorthogonal sequences and $\D$ be a dense subspace in $\Hil$ such that $D_\varphi \cup D_\psi \subseteq \D \subseteq D(\varphi) \cap D(\psi)$. Then the following statements are equivalent:
\par
(i) $(\F_\varphi,\F_\psi)$ is a $\D$-quasi basis.
\par
(ii) For any ONB $\F_e=\{e_n\}$ in $\Hil$, $\F_\varphi$ is a generalized Riesz system with a natural constructing pair $(\F_e,T)$ satisfying $D(T^{\ast}) \cap D(T^{-1}) \supseteq \D$.
\par
(iii) For any ONB $\F_e=\{e_n\}$ in $\Hil$, $\F_\psi$ is a generalized Riesz system with a natural constructing pair $(\F_e,K)$ satisfying $D(K^{\ast}) \cap D(K^{-1}) \supseteq \D$.\\
If (i) holds, then we can take $(\overline{T_{\bm{e},\psi}\lceil_\D})^{-1}$ and $(\overline{T_{\bm{e},\varphi}\lceil_\D})^{-1}$ as $T$ in (ii) and $K$ in (iii), respectively.\\
\end{cor}

By Theorem 3.2, if $(\F_\varphi,\F_\psi)$ is a $(\D,\E)$-quasi basis, then, for any ONB $\F_e= \{ e_n \}$, $(\overline{T_{\bm{e},\psi}\lceil_\E})^{-1}$ and $(\overline{T_{\bm{e},\varphi}\lceil_\D})^{\ast}$ are constructing operators for the generalized Riesz system $\F_\varphi$ such that $(\overline{T_{\bm{e},\psi}\lceil_\E})^{-1} \subseteq (\overline{T_{\bm{e},\varphi}\lceil_\D})^{\ast}$, and $(\overline{T_{\bm{e},\varphi}\lceil_\D})^{-1}$ and $(\overline{T_{\bm{e},\psi}\lceil_\E})^{\ast}$ are constructing operators for the generalized Riesz system $\F_\psi$ such that $(\overline{T_{\bm{e},\varphi}\lceil_\D})^{-1} \subseteq (\overline{T_{\bm{e},\psi}\lceil_\E})^{\ast}$.\\

{\bf Remark.} For a biorthogonal pair $(\F_\varphi,\F_\psi)$ it is clear that $D_\psi \subseteq D(\varphi)$ and $D_\varphi \subseteq D(\psi)$. What is not clear is whether $D_\varphi \subseteq D(\varphi)$ and $D_\psi \subseteq D(\psi)$. For this reason it may be more convenient to work, in some concrete cases, with $(\D,\E)$-quasi bases rather than with $\D$-quasi bases.

\vspace{2mm}

Let $\F_\varphi$ be a generalized Riesz system with constructing pair $(\F_e,T)$. We discuss now when there exists a sequence $\F_\psi$ in $\Hil$ and subspaces $\D$ and $\E$ in $\Hil$ such that $\F_\varphi$ and $\F_\psi$ are biorthogonal and define a $(\D,\E)$-quasi basis:

\vspace{2mm}

\begin{prop} Let $\F_\varphi$ be a generalized Riesz system with a constructing pair $(\F_e,T)$. Then $(\F_\varphi,\F_\psi^T)$ is a $(D(T^{\ast}),D(T^{-1}))$-quasi basis and
$T= \left(T_{\bm{e},\psi^T}\lceil_{D(T^{-1})} \right)^{-1}$, $(T^{-1})^{\ast}= \left(T_{\bm{e},\varphi}\lceil_{D(T^{\ast})} \right)^{-1}$.
\end{prop}
\begin{proof} It is clear that  $(\F_\varphi,\F_\psi^T)$ is a $(D(T^{\ast}),D(T^{-1}))$-quasi basis. Furthermore, since $Te_n=\varphi_n$, $n=0,1, \cdots$, we have
\be
T_{\varphi,\bm{e}} \subseteq T ,\nonumber
\en
which implies that
\be
T^{\ast} \subseteq T_{\bm{e},\varphi}. \nonumber
\en
Hence we have
\be
T^{\ast} =T_{\bm{e},\varphi}\lceil_{D(T^{\ast})}. \nonumber
\en
Thus we have
\be
(T^{\ast})^{-1}=\left( T_{\bm{e},\varphi}\lceil_{D(T^{\ast})} \right)^{-1}. \nonumber
\en
Since $(T^{-1})^\ast e_n=\psi_n^T$, $n=0,1, \cdots$, we can similarly show $T= \left(T_{\bm{e},\psi^T}\lceil_{D(T^{-1})} \right)^{-1}$.
This completes the proof.
\end{proof}
Next we consider when there exists a subspace $\D$ in $\Hil$ such that $(\F_\varphi,\F_\psi^T)$ is $\D$-quasi basis.\\
\par
\begin{prop} Let $\F_\varphi$ be a generalized Riesz system with constructing pair $(\F_e,T)$. Suppose that $\F_e \subset D(T^{\ast}T) \cap D(T^{-1}(T^{-1})^{\ast})$.
 Then $(\F_\varphi,\F_\psi^T)$ is a $(D(T^{\ast}) \cap D(T^{-1}))$-quasi basis and $T= \left( \overline{T_{\bm{e},\psi^T}\lceil_{D(T^{\ast}) \cap D(T^{-1})}} \right)^{-1}$, $(T^{-1})^{\ast}= \left( \overline{T_{\bm{e},\varphi}\lceil_{D(T^{\ast}) \cap D(T^{-1})}} \right)^{-1}$.
\end{prop}
\begin{proof} We denote for simplicity $\psi^T$ by $\psi$. At first, we show that $D(T^{-1}) \cap D(T^{\ast})$ is a core for $T^{-1}$. Take an arbitrary $x\in D(T)$. Let $|T|=\int_0^\infty \lambda d E_T(\lambda)$ be the spectral resolution of the  absolute $|T| := (T^{\ast}T)^{1/2}$ of $T$. Then we have $TE_T(n) x \in D(T^{\ast}) \cap D(T^{-1})$, $n=0,1, \cdots$ and $\lim_{n \rightarrow \infty} TE_T(n) x=Tx$. Furthermore, take an arbitrary $y\in D(T^{-1})$. Then $y=Tx$ for some $x\in D(T)$ and we have $\lim_{n \rightarrow \infty} TE_T(n) x=Tx=y$ and $\lim_{n \rightarrow \infty}T^{-1}(TE_T(n) x)=\lim_{n \rightarrow \infty}E_T(n)x=x=T^{-1}y.$ Thus $D(T^{-1}) \cap D(T^{\ast})$ is a core for $T^{-1}$.

At second, we show that $D(T^{-1}) \cap D(T^{\ast})$ is a core for $T^{\ast}$. Take an arbitrary $y\in D(T^{\ast})$. Let $|T^{\ast}|=\int_0^\infty \lambda d E_{T^{\ast}}(\lambda)$ be the spectral resolution of the  absolute $|T^{\ast}| := (TT^{\ast})^{1/2}$ of $T^{\ast}$. Then it follows that $E_{T^\ast}(n)y= T(T^\ast|T^\ast|^{-2} E_{T^\ast}(n)y) \in D(T^{-1}) \cap D(T^\ast)$, $n=0,1, \cdots$, $\lim_{n \rightarrow \infty} E_{T^{\ast}}(n) y =y$ and $\lim_{n \rightarrow \infty} T^{\ast}E_{T^{\ast}}(n) y=T^{\ast}y$. Thus $D(T^{-1}) \cap D(T^{\ast})$ is a core for $T^{\ast}$.

At third, we show that $D_\varphi \subseteq D(T^{-1})\cap D(T^{\ast}) \subseteq D(\varphi) \cap D(\psi)$ and $D_\psi \subseteq D(T^{-1})\cap D(T^{\ast}) \subseteq D(\varphi ) \cap D(\psi)$. It is clear that $\varphi_n=Te_n \in D(T^{-1})$. Furthermore, since $\F_{e} \subseteq D(T^\ast T)$, we have
$$
\ip{Tx}{\varphi_n}
=\ip{Tx}{Te_n} \nonumber
= \ip{x}{T^{\ast}Te_n} $$
for all $x\in D(T)$. Hence we have $\varphi_n \in D(T^{\ast})$. Thus $D_\varphi \subseteq D(T^{-1}) \cap D(T^{\ast})$. And since $\psi_n=(T^{-1})^{\ast} e_n (=(T^{\ast})^{-1}e_n)$, we have $\psi_n\in D(T^{\ast})$. Furthermore, since $\F_{e} \subseteq D(T^{-1}(T^{-1})^\ast)$, we have
$$
\ip{(T^{-1})^{\ast}y}{ \psi_n}
= \ip{(T^{-1})^{\ast} y}{ (T^{-1})^{\ast}e_n}
= \ip{y}{ T^{-1}(T^{-1})^{\ast} e_n} $$
for all $y\in D((T^{-1})^{\ast})$. Hence we have $\psi_n \in D(T^{-1})$. Thus $D_\psi \subseteq D(T^{-1}) \cap D(T^{\ast})$. We show $D(T^{-1}) \cap D(T^\ast) \subseteq D(\varphi) \cap D(\psi)$. Indeed, take an arbitrary $y\in D(T^{-1}) \cap D(T^{\ast})$. Since
$$
\sum_{k=0}^{\infty} |\ip{y}{\varphi_k}|^2
= \sum_{k=0}^\infty |\ip{y}{Te_k}|^2
= \sum_{k=0}^\infty |\ip{T^{\ast}y}{e_k}|^2
= \| T^{\ast}y\|^2
$$
and
$$
\sum_{k=0}^{\infty} |\ip{y}{\psi_k}|^2
= \sum_{k=0}^\infty |\ip{T^{-1}y}{e_k}|^2
= \| T^{-1}y\|^2, $$
 we have $y\in D(\varphi) \cap D(\psi)$.

Finally, we show that $(\F_\varphi,\F_\psi^T)$ is a $(D(T^{\ast})\cap D(T^{-1}))$-quasi basis and $T= \left( \overline{T_{\bm{e},\psi}\lceil_{D(T^{\ast}) \cap D(T^{-1})}} \right)^{-1}$, $(T^{-1})^{\ast}= \left( \overline{T_{\bm{e},\varphi}\lceil_{D(T^{\ast}) \cap D(T^{-1})}} \right)^{-1}$. Since
\begin{eqnarray}
\sum_{k=0}^\infty \ip{x}{ \varphi_k}\ip{\psi_k}{y}
&=& \sum_{k=0}^\infty \ip{x}{Te_k}\ip{(T^{-1})^{\ast}e_k}{y} \nonumber \\
&=& \sum_{k=0}^\infty \ip{T^{\ast} x}{ e_k}\ip{e_k}{T^{-1} y} \nonumber \\
&=& \ip{T^{\ast}x}{T^{-1}y} \nonumber \\
&=& \ip{x}{y} \nonumber
\end{eqnarray}
for all $x,y \in D(T^{\ast}) \cap D(T^{-1})$, it follows that
$(\F_\varphi,\F_\psi^T)$ is a $(D(T^{\ast})\cap D(T^{-1}))$-quasi basis. Furthermore since $T^{-1} \subseteq T_{\bm{e},\psi}$ and  $D(T^{-1}) \cap D(T^{\ast})$ is a core for $T^{-1}$, we have
\be
T^{-1}=\overline{T^{-1}\lceil_{D(T^{\ast}) \cap D(T^{-1})}}= \overline{T_{\bm{e},\psi}\lceil_{D(T^{\ast}) \cap D(T^{-1})}}, \nonumber
\en
which implies that $T=(\overline{T_{\bm{e},\psi}\lceil_{D(T^{\ast}) \cap D(T^{-1})}})^{-1}$. Furthermore since $T_{\varphi,\bm{e}} \subseteq T$ and $D(T^{-1}) \cap D(T^{\ast})$ is a core for $T^{\ast}$, we have
\be
T^{\ast}= \overline{T^{\ast} \lceil_{D(T^{\ast}) \cap D(T^{-1})}}=\overline{T_{\bm{e},\varphi}\lceil_{D(T^{\ast}) \cap D(T^{-1})}}, \nonumber
\en
which implies that $(T^{\ast})^{-1}=(\overline{T_{\bm{e},\varphi}\lceil_{D(T^{\ast}) \cap D(T^{-1})}})^{-1}$. This completes the proof.
\end{proof}

\section{Physical operators constructed from $(\D,\E)$-quasi bases}
In this section, extending what was discussed recently for instance in \cite{bit2013,hiro1,bb2017},  we investigate some physical operators constructed from $(\D,\E)$-quasi bases. Let $(\F_\varphi,\F_\psi)$ be a $(\D,\E)$-quasi basis. As shown in Theorem 3.2, $F_\varphi$ is a generalized Riesz system with constructing pairs $(\F_e, (\overline{T_{\bm{e},\psi} \lceil_\E})^{-1})$ and $(\F_e,(T_{\bm{e},\varphi} \lceil_\D)^\ast)$ for any ONB $\F_e =\{ e_n \}$ such that $(\overline{T_{\bm{e},\psi} \lceil_\E})^{-1} \subseteq (T_{\bm{e},\varphi} \lceil_\D)^\ast$, and $\{ \psi_n \}$ is a generalized Riesz system with constructing pairs $(\F_e,(\overline{T_{\bm{e},\varphi}} \lceil_\D)^{-1})$ and $(\F_e,(T_{\bm{e},\psi\lceil_\D})^\ast)$ such that $(\overline{T_{\bm{e},\varphi} \lceil_\D})^{-1} \subseteq (T_{\bm{e},\psi}\lceil_\E)^\ast$. Here we put, to keep the notation simple,
\begin{eqnarray}
T&=& (\overline{T_{\bm{e},\psi} \lceil_\E})^{-1} \;\;\; {\rm or} \;\;\; (T_{\bm{e},\varphi} \lceil_\D)^\ast, \nonumber \\
K&=& (\overline{T_{\bm{e},\varphi} \lceil_\D})^{-1} \;\;\; {\rm or} \;\;\; (T_{\bm{e},\psi}\lceil_\E)^\ast . \nonumber
\end{eqnarray}
For a generalized Riesz system $\F_\varphi$ with constructing pair $(\F_e,T)$ we can define a non-self-adjoint Hamiltonian $H_\varphi^{\bm{\alpha}} := TH_{\bm{e}}^{\bm{\alpha}}T^{-1}$, a generalized lowering operator $A_\varphi^{\bm{\alpha}} := TA_{\bm{e}}^{\bm{\alpha}}T^{-1}$ and a generalized raising operator $B_\varphi^{\bm{\alpha}} := TB_{\bm{e}}^{\bm{\alpha}}T^{-1}$. Similarly, for a generalized Riesz system $\{ \psi_n \}$ with a constructing pair $(\F_e,K)$ we define a non-self-adjoint Hamiltonian $H_\psi^{\bm{\alpha}} := KH_{\bm{e}}^{\bm{\alpha}}K^{-1}$, a generalized lowering operator $A_\psi^{\bm{\alpha}} := KA_{\bm{e}}^{\bm{\alpha}}K^{-1}$ and a generalized raising operator $B_\psi^{\bm{\alpha}} := KB_{\bm{e}}^{\bm{\alpha}}K^{-1}$. But we don't know whether these operators are even densely defined or not.
Suppose that $D_\varphi$ is dense in $\Hil$. Then, since $H_{\varphi}^{\bm{\alpha}} \varphi_n
= \alpha_n \varphi_n$, $A_{\varphi}^{\bm{\alpha}} \varphi_n = \alpha_n \varphi_{n-1}$ $(0 \;{\rm if}\; n=0)$ and
$B_{\varphi}^{\bm{\alpha}} \varphi_n= \alpha_{n+1} \varphi_{n+1}$,
it is clear that $H_{\varphi}^{\bm{\alpha}}$, $A_{\varphi}^{\bm{\alpha}}$ and $B_{\varphi}^{\bm{\alpha}}$ are densely defined, but since $D_\psi$ is not necessarily dense in $\Hil$, the operators  $H_{\psi}^{\bm{\alpha}}$, $A_{\psi}^{\bm{\alpha}}$ and $B_{\psi}^{\bm{\alpha}}$ need not being densely defined.
 Therefore,  we first investigate when $\D_\varphi$ or $\D_\psi$ are dense in $\Hil$ under the assumption that $(\F_\varphi ,\F_\psi)$ is a $(\D,\E)$-quasi basis.\\

{Before going forth, we shortly discuss an example which is the leading model for the objects we are dealing with and which allows an explicit computation of all involved operators.
 	
 	\vspace{2mm}
 	
 	{\bf Example 3:}-- Let $H_0=p^2+x^2$ be  the self-adjoint Hamiltonian introduced in Example 2 above, and let $T$ be the following multiplication operator: $(Tf)(x)=(1+x^2)f(x)$, for all functions $f(x)\in D(T)=\{g(x)\in\Lc^2(\mathbb{R}): \, (1+x^2)g(x)\in\Lc^2(\mathbb{R})\}$. $T$ is an unbounded self-adjoint operator, invertible with bounded inverse $T^{-1}$.

 As seen in \eqref{eqn_Hzero}, $H_0$ has the form $H_{\bm{e}}^{\bm{\alpha}}$ where ${\bm{\alpha}}= \{2n+1,\, n\in {\mb N}\}$ and $\{e_n\}$ is the orthonormal basis constructed from the Hermite polynomials. To simplify notations, we will omit here explicit reference to ${\bm{\alpha}}$.

  If we identify $K$ with $T^{-1}$, 
  straightforward computations show that
 	$$
 	H_\varphi=p^2+V_\varphi(x)+\frac{4ix}{1+x^2}\,p, \qquad H_\psi=p^2+V_\psi(x)-\frac{4ix}{1+x^2}\,p,
 	$$
 	where $V_\varphi(x)=x^2+2\frac{(1-3x^2)}{(1+x^2)^2}$ and  $V_\psi(x)=x^2-\frac{2}{1+x^2}$. Notice that, because of the relation between $T$ and $K$, $H_\varphi=H_\psi^*$, even if this is not evident from our explicit formulas. From a physical point of view both $H_\varphi$ and $H_\psi$ can be seen as a modified version of the harmonic oscillator where an extra potential is added, going to zero as $x^{-2}$, and the manifestly non self-adjoint terms $\pm \frac{4ix}{1+x^2}\,p$ appear. These Hamiltonians can be factorized as follows: $H_\varphi=2B_\varphi A_\varphi+\1$ and $H_\psi=2B_\psi A_\psi+\1$, where
 	$$
 	A_\varphi=\frac{1}{\sqrt{2}}\left(x-\frac{2x}{1+x^2}+ip\right), \quad B_\varphi=\frac{1}{\sqrt{2}}\left(x+\frac{2x}{1+x^2}-ip\right),
 	$$
 	while
 	$$
 	A_\psi=\frac{1}{\sqrt{2}}\left(x+\frac{2x}{1+x^2}+ip\right), \quad B_\psi=\frac{1}{\sqrt{2}}\left(x-\frac{2x}{1+x^2}-ip\right).
 	$$
 All these operators formally collapse to $c=\frac{1}{\sqrt{2}}\left(x+ip\right)$ or to $c^\dagger=\frac{1}{\sqrt{2}}\left(x-ip\right)$ for large $x$. It is also interesting to observe that $B_\varphi=A_\psi^*$ and $A_\varphi=B_\psi^*$
 	
 	The two vacua of $A_\varphi$ and $A_\psi$, corresponding to the lower eigenvectors of $H_\varphi$ and $H_\psi$ respectively, can be easily obtained by solving the differential equations $A_\varphi\varphi_0(x)=0$ and $A_\psi\psi_0(x)=0$. The solutions we find in this way coincide with those we find introducing
 	$$
 	\varphi_n(x)=(Te_n)(x)=\frac{1}{\sqrt{2^n\,n!\,\pi^{1/2}}}(1+x^2)H_n(x)e^{-x^2/2},
 	$$
 	and
 	$$
 	\varphi_n(x)=(Ke_n)(x)=\frac{1}{\sqrt{2^n\,n!\,\pi^{1/2}}}\frac{H_n(x)}{1+x^2}\,e^{-x^2/2},
 	$$
 	see Example 2. Incidentally, it is clear that $e_n(x)\in D(T)$. Of course, $e_n(x)\in D(K)$ since $D(K)=\Lc^2(\mathbb{R})$.
 	
 	The last point we want to consider here concerns the density of  $\D_\varphi$ and $\D_\psi$ in $\Lc^2(\mathbb{R})$. More concretely, we will check that $\F_\varphi$ is total in $D(T)$ and that  $\F_\psi$ is total in $D(K)=\Lc^2(\mathbb{R})$. In fact, let $f(x)\in D(T)$ be such that $\left<f,\varphi_n\right>=0$ for all $n$. Hence $0= \left<f,\varphi_n\right>= \left<Tf,e_n\right>$, so that $Tf=0$ and, since $Tf\in D(K)$, $f(x)=0$ a.e. in $\mathbb{R}$. Similarly we can prove that, if $g(x)\in\Lc^2(\mathbb{R})$ is such that $\left<g,\psi_n\right>=0$ for all $n$, then $g(x)=0$ a.e. in $\mathbb{R}$.
 	
 	\medskip
 We come now back to investigate more general situations.

 }

\par
{\bf Proposition 4.1.} {\it Suppose that $(\F_\varphi ,\F_\psi)$ is a $(\D,\E)$-quasi basis. Then, we have the following statements.
\par
(1) $D_{\varphi}^\perp \subseteq D(\varphi)$, where $D_\varphi^\perp$ is an orthogonal complement of $D_\varphi$ in $\Hil$.
\par
(2) If $\D \cap D_{\varphi}^\perp$ is dense in $D_{\varphi}^\perp$, then $D_\varphi$ is dense in $\Hil$.\\
Similar results hold for $\F_\psi$.}\\
\par

\begin{proof}
(1) For $x\in D_\varphi^\perp$, we have
\begin{eqnarray}
\ip{T_{\varphi,\bm{e}}e_n}{x}
=\ip{\varphi_n}{x}=0 \nonumber
\end{eqnarray}
for any ONB $\F_{\bm{e}}$ in $\Hil$ and $n=0,1, \cdots$. Since $\F_e$ is a core for $\bar{T}_{\varphi,\bm{e}}$ by Lemma 2.2, we have $x\in D(T_{\varphi, \bm{e}}^\ast)=D(T_{\bm{e},\varphi})=D(\varphi)$.
\par
(2) For any $x\in D_\varphi^\perp$, there exists a sequence $\{ x_n \} \subseteq \D \cap D_\varphi^\perp$ such that $\lim_{n \rightarrow \infty} x_n =x$. Since $(\F_\varphi ,\F_\psi)$ is a $(\D,\E)$-quasi basis, we have
\begin{eqnarray}
\ip{x}{y}
&=& \lim_{n \rightarrow \infty} \ip{x_n}{y} \nonumber \\
&=& \lim_{n \rightarrow \infty} \sum_{k=0}^\infty \ip{x_n}{\varphi_k}\ip{\psi_k}{y}=0 \nonumber
\end{eqnarray}
for all $y\in \E$. Hence we have $x=0$. Thus $D_\varphi$ is dense in $\Hil$.
\end{proof}

{\bf Proposition 4.2.} {\it Let $( \F_\varphi,\F_\psi)$ be a biorthogonal pair such that $D(\varphi)$ and $D(\psi)$ are dense in $\Hil$. Then we have the following
\par
(1) $( \F_\varphi,\F_\psi)$ is a $(D(\varphi),\E)$-quasi basis for some dense subspace $\E$ in $\Hil$ such that $D_\varphi \subseteq \E \subseteq D(\psi)$ if and only if $D_\varphi$ is dense in $\Hil$. If this is true, $(\F_\varphi,\F_\psi)$ is a $(D(\varphi),D_\varphi)$-quasi basis.
\par
(2) $(\F_\varphi,\F_\psi)$ is a $(\D,D(\psi))$-quasi basis for some dense subspace $\D$ in $\Hil$ such that $D_\psi \subseteq \D \subseteq D(\varphi)$ if and  only if $D_\psi$ is dense in $\Hil$. If this is true, $(\F_\varphi,\F_\psi)$ is a $(D_\psi,D(\psi))$-quasi basis.}


\begin{proof}  (1) Suppose that $(\F_\varphi,\F_\psi)$ is a $(D(\varphi),\E)$-quasi basis for some dense subspace $\E$ in $\Hil$ such that $D_\varphi \subseteq \E \subseteq D(\psi)$. Take an arbitrary
$x\in D_\varphi^\perp$. By Proposition 4.1, (1) we have $x\in D(\varphi)$. Since $( \{ \varphi_n \}, \{ \psi_n \})$ is a $(D(\varphi),\E)$-quasi basis, we have
\begin{eqnarray}
\ip{x}{y}
= \sum_{k=0}^\infty \ip{x}{\varphi_k}\ip{\psi_k}{y}=0 \nonumber
\end{eqnarray}
for all $y\in \E$, which implies that $x=0$. Hence $D_\varphi$ is dense in $\Hil$. \\
Conversely suppose that $\D_\varphi$ is dense in $\Hil$.
Then we show that $(\F_\varphi,\F_\psi)$ is a $(D(\varphi),D_\varphi)$-quasi basis. Indeed, take arbitrary $x\in D(\varphi)$ and $y\in D_\varphi$. Then, $y= \sum_{j=0}^n \alpha_j \varphi_j$ for some $\alpha_j \in \C$, $j=0,1, \cdots , n$, and we have
\begin{eqnarray}
\sum_{k=0}^\infty \ip{x}{\varphi_k}\ip{\psi_k}{y}
&=& \sum_{k=0}^\infty \ip{x}{T_{\varphi,\bm{e}}e_k}\ip{T_{\psi,\bm{e}}e_k}{y} \nonumber \\
&=& \ip{T_{\bm{e},\varphi}x}{ T_{\bm{e},\psi}y} \nonumber \\
&=& \sum_{j=0}^n \bar{\alpha}_j \ip{T_{\bm{e},\varphi}x}{ T_{\bm{e},\psi}\varphi_j} \nonumber \\
&=& \sum_{j=0}^n \bar{\alpha}_j \ip{x}{ T_{\varphi,\bm{e}}e_j} \nonumber \\
&=& \ip{x}{\sum_{j=0}^n \alpha_j \varphi_j} \nonumber \\
&=& \ip{x}{y} . \nonumber
\end{eqnarray}
\par
(2) This is shown similarly to (1).

\end{proof}

Suppose that $(\F_\varphi ,\F_\psi)$ is a $(\D,\E)$-quasi basis. Let $\bm{r} := \{ r_n \} \subset \R$; $1 \leq r_n$, $n=0,1, \cdots$ and we put
\begin{eqnarray}
\varphi_r &:=& \{ r_n \varphi_n \} , \nonumber \\
\psi_{\frac{1}{r}}
&:=& \left\{ \frac{1}{r_{n}} \psi_n \right\} . \nonumber
\end{eqnarray}
Then, $(\varphi_r ,\psi_{\frac{1}{r}})$ is a biorthogonal pair satisfying
\begin{eqnarray}
D_{\psi_r}&=&D_\psi \subseteq D(\varphi_r)  \subseteq D(\varphi), \nonumber \\
D_{\varphi_r} &=&D_\varphi \subseteq \E \subseteq D(\psi) \subseteq D(\psi_{\frac{1}{r}}) , \nonumber
\end{eqnarray}
where $$D(\varphi_r) := \left\{ x \in \Hil ; \sum_{k=0}^\infty r_k^2 |\ip{x}{\varphi_k}|^2<\infty
\right\} \mbox{\; and \;} D(\psi_{\frac{1}{r}}) := \left\{ x \in \Hil ; \sum_{k=0}^\infty \frac{1}{r_k^2} |\ip{x}{\psi_k}|^2<\infty \right\}.$$\\
Then we have the following \\
\par
{\bf Proposition 4.3.} {\it Suppose that $(\F_\varphi,\F_\psi)$ is a $(\D,\E)$-quasi basis and
there exists a sequence $\bm{r} := \{ r_n \} \subset \R$ such that $1\leq r_n$, $n=0,1, \cdots$ and $D(\varphi_r) \subseteq \D$ and $D(\varphi_r)$ is dense in $\Hil$. Then, $D_\varphi$ is dense in $\Hil$ and $(\F_\varphi,\F_\psi)$ is a $(D(\varphi),D_\varphi)$-quasi basis.}\\
\par
\begin{proof} Since $D(\varphi_r) \subseteq \D$, it follows that $(\varphi_r ,\psi_{\frac{1}{r}})$ is a $(D(\varphi_r),\E)$-quasi basis, which implies by Proposition 4.2 that $D_{\varphi_r}=D_\varphi$ is dense in $\Hil$.
\end{proof}

We next consider the case that $D_\varphi$ and $D_\psi$ are not necessarily dense in $\Hil$.\\
\par
{\bf Proposition 4.4.} {\it Suppose that $(\F_\varphi,\F_\psi)$ is a $(\D,\E)$-quasi basis. Then there exists an ONB $\F_f := \{ f_n \}$ in $\Hil$ such that $\overline{T_{\bm{f},\varphi} \lceil_\D}$ is a positive self-adjoint operator in $\Hil$ and $(\F_{\bm{f}},\overline{T_{\bm{f},\varphi} \lceil_\D})$ is a constructing pair for the generalized Riesz system $\F_\varphi$. Furthermore, $(\F_{\bm{f}}, (\overline{T_{\bm{f},\varphi} \lceil_\D})^{-1})$ is a constructing pair for the generalized Riesz system $\F_\psi$.\\
}
\par
\begin{proof} By Theorem 3.2, $(\overline{T_{\bm{e},\varphi} \lceil_\D})^\ast$ is a constructing operator for the generalized Riesz system $\F_\varphi$ and any ONB $\F_e = \{ e_n \}$ in $\Hil$. Let $\overline{T_{\bm{e},\varphi} \lceil_\D}=U|\overline{T_{\bm{e},\varphi} \lceil_\D}|$ be the polar decomposition of $\overline{T_{\bm{e},\varphi} \lceil_\D}$. Since $\overline{T_{\bm{e},\varphi} \lceil_\D}$ has a densely defined inverse, $U$ is a unitary operator on $\Hil$. Here we put $f_n =U^\ast e_n$, $n=0,1, \cdots$. Then it follows that $\{ f_n \}$ is an ONB in $\Hil$ and
\begin{eqnarray}
|\overline{T_{\bm{e},\varphi} \lceil_\D}| f_n
= |\overline{T_{\bm{e},\varphi} \lceil_\D}|U^\ast e_n
=(T_{\bm{e},\varphi} \lceil_\D)^\ast e_n
=\varphi_n , \;\;\; n=0,1, \cdots , \nonumber
\end{eqnarray}
which implies that $(\F_{\bm{f}}, |\overline{T_{\bm{e},\varphi} \lceil_\D}|)$ is a constructing pair for $\F_\varphi$. Hence,
\begin{eqnarray}
T_{\varphi,\bm{f}} \subseteq |\overline{T_{\bm{e},\varphi} \lceil_\D}|
\subseteq T_{\bm{f},\varphi} , \nonumber
\end{eqnarray}
and so $\overline{T_{\bm{f},\varphi} \lceil_\D} =|\overline{T_{\bm{e},\varphi} \lceil_\D}|$. This completes the proof.
\end{proof}

\par
Similarly we have the following\\
\par
{\bf Proposition 4.5.} {\it Suppose that $(\F_\varphi,\F_\psi)$ is a $(\D,\E)$-quasi basis. Then there exists an ONB $\F_g := \{ g_n \}$ in $\Hil$ such that $\overline{T_{\bm{g},\psi} \lceil_\E}$ is a positive self-adjoint operator in $\Hil$ and $(\F_{\bm{g}},\overline{T_{\bm{g},\psi} \lceil_\E})$ is a constructing pair for the generalized Riesz system $\F_\psi$. Furthermore, $(\F_{\bm{g}}, (\overline{T_{\bm{g},\psi} \lceil_\E})^{-1})$ is a constructing pair for the generalized Riesz system $\F_\varphi$.\\}
\par
We now consider a CCR-algebra-like structure for non-self-adjoint Hamiltonians, generalized lowering and raising operators by taking a good domain for their operators. For that the notion of unbounded operator algebras is relevant, \cite{schm, bagrus, ct_heisen}. Let $\D$ be a dense subspace in a Hilbert space $\Hil$. We denote by $\Lc(\D)$ the set of all linear operators from $\D$ to $\D$. Then $\Lc (\D)$ is an algebra equipped with the usual operations: $X+Y$, $\alpha X$ and $XY$.\\
\par

{\bf Theorem 4.6.} {\it Suppose that $(\F_\varphi,\F_\psi)$ is a $(\D,\E)$-quasi basis, and $\F_{\bm{f}} = \{ f_n \}$ and $\F_{\bm{g}}= \{ g_n \}$ in Proposition 4.4 and Proposition 4.5. Here we denote by $T_\varphi$ the constructing operator $\overline{T_{\bm{f},\varphi}\lceil_\D }$ of $\F_\varphi$ and $T_\psi$ the constructing operator $\overline{T_{\bm{g},\psi}\lceil_\E }$ of $\F_\psi$. Then we have the following
\par
(1) If $H_{\bm{f}}^{\bm{\alpha}} \D \subseteq \D$ for some ${\bm{\alpha}} = \{ \alpha_n \} \subset \C$, then the linear span of $T_\varphi \D$ is dense in $\Hil$ and the non-self-adjoint Hamiltonian $T_\varphi H_{\bm{f}}^{\bm{\alpha}}T_\varphi^{-1}$ for $\F_\varphi$ is contained in ${\cal L}(T_\varphi \D)$.
\par
(2) If $H_{\bm{g}}^{\bm{\alpha}} \E \subseteq \E$ for some ${\bm{\alpha}} = \{ \alpha_n \} \subset \C$, then the linear span of $T_\psi \E$ is dense in $\Hil$ and the non-self-adjoint Hamiltonian $T_\psi^{-1} H_{\bm{g}}^{\bm{\alpha}}T_\psi$ for $\F_\psi$ is contained in ${\cal L}(T_\psi\E)$.\\
Here $H_{\bm{f}}^{\bm{\alpha}}$ and $H_{\bm{g}}^{\bm{\alpha}}$ are the standard Hamiltonians for the ONB $\F_{\bm{f}}$ and $\F_{\bm{g}}$, respectively.\\}
\par
\begin{proof} (1) Since $\D$ is a core for $T_\varphi$ and $T_\varphi$ has the inverse, $T_\varphi \D$ is dense in $\Hil$. By assumption, it is clear that $T_\varphi H_{\bm{f}}^{\bm{\alpha}} T_\varphi^{-1} \in \Lc (T_\varphi \D)$.
\par
(2) This is shown similarly to (1).
\end{proof}
\par
Next, to consider the generalized lowering and raising operators defined by $(\D,\E)$-quasi bases, we assume that
\begin{eqnarray}
0 \leq \alpha_0 < \alpha_n <\alpha_{n+1} \;\;\; {\rm and } \;\;\; \alpha_{n+1} \leq \alpha_n +r , \;\; n=1, \cdots, \;\;\; {\rm for \; some} \; r>0.
\end{eqnarray}
Then we have the following\\
\par
{\bf Theorem 4.7.} {\it Suppose that $(\F_\varphi,\F_\psi)$ is a $(\D,\E)$-quasi basis, and $T_\varphi$, $T_\psi$, $\F_{\bm{f}} = \{ f_n \}$ and $\F_{\bm{g}}= \{ g_n \}$ as in Theorem 4.6. Then we have the following statements.
\par
(1) Suppose that $D^\infty(H_{\bm{f}}^{\bm{\alpha}})  := \cap_{n \in N} D((H_{\bm{f}}^{\bm{\alpha}})^n) \subseteq \D$ and $T_{\bm{f},\varphi}D^\infty (H_{\bm{f}}^{\bm{\alpha}})$ is dense in $\Hil$. Then $(\F_{\bm{f}}, T_\varphi^0 :=\overline{T_{\bm{f},\varphi} \lceil_{D^\infty (H_{\bm{f}}^{\bm{\alpha}})}})$ is a constructing pair for $\F_\varphi$ and the non-self-adjoint Hamiltonian $H_\varphi^0 := T_\varphi^0 H_{\bm{f}}^{\bm{\alpha}} (T_\varphi^0)^{-1}$ for $\F_\varphi$, the generalized lowering operator $A_\varphi^0 := T_\varphi^0 A_{\bm{f}}^{\bm{\alpha}} (T_\varphi^0)^{-1}$ for $\F_\varphi$ and the generalized raising operator $B_\varphi^0 := T_\varphi^0 B_{\bm{f}}^{\bm{\alpha}} (T_\varphi^0)^{-1}$ for $\F_\varphi$ are contained in $\Lc ( T_\varphi^0 D^\infty (H_{\bm{f}}^{\bm{\alpha}}))$.
\par
(2) Suppose that $D^\infty(H_{\bm{g}}^{\bm{\alpha}})   \subseteq \E$ and $T_{\bm{g},\psi}D^\infty (H_{\bm{g}}^{\bm{\alpha}})$ is dense in $\Hil$. Then $(\F_{\bm{g}}, T_\psi^0 :=\overline{T_{\bm{g},\psi} \lceil_{D^\infty (H_{\bm{g}}^{\bm{\alpha}})}})$ is a constructing pair for $\F_\psi$ and the non-self-adjoint Hamiltonian $H_\psi^0 := T_\psi^0 H_{\bm{g}}^{\bm{\alpha}} (T_\psi^0)^{-1}$ for $\F_\psi$, the generalized lowering operator $A_\psi^0 := T_\psi^0 A_{\bm{g}}^{\bm{\alpha}} (T_\psi^0)^{-1}$ for $\F_\psi$ and the generalized raising operator $B_\psi^0 := T_\psi^0 B_{\bm{g}}^{\bm{\alpha}} (T_\psi^0)^{-1}$ for $\F_\psi$ are contained in $\Lc ( T_\psi^0 D^\infty (H_{\bm{g}}^{\bm{\alpha}}))$.\\
}
\par
\begin{proof} At first, we show that $(\F_{\bm{f}},T_\varphi^0)$ is a constructing pair for $\F_\varphi$. Since $D(T_\varphi^0) \supseteq D^\infty (H_{\bm{f}}^{\bm{\alpha}}) \supseteq \F_{\bm{f}}$, $T_\varphi^0$ is a densely defined closed operator in $\Hil$. Furthermore, since $T_\varphi^0 \subseteq T_\varphi = \overline{T_{\bm{f},\varphi}\lceil_\D}$ and $T_\varphi$ has the inverse, $T_\varphi^0$ has the inverse. By assumption, we have
\begin{eqnarray}
D((T_\varphi^0)^{-1})
\supseteq T_\varphi^0 D(T_\varphi^0)
\supseteq T_\varphi^0 D^\infty (H_{\bm{f}}^{\bm{\alpha}})
=T_{\bm{f},\varphi } D^\infty ( H_{\bm{f}}^{\bm{\alpha}}), \nonumber
\end{eqnarray}
which implies that $T_\varphi^0$ has a densely defined inverse. Furthermore, we have the following
\begin{eqnarray}
T_\varphi^0 f_n
=T_\varphi f_n
=\varphi_n , \;\;\; n=0,1, \cdots . \nonumber
\end{eqnarray}
Hence we have $(\F_\varphi,T_\varphi^0)$ is a constructing pair for $\F_\varphi$.
\par
Next we consider the non-self-adjoint Hamiltonian $H_\varphi^0$ for $\F_\varphi$, the generalized lowering operator $A_\varphi^0 $ for $\F_\varphi$ and the generalized raising operator for $B_\varphi^0 $ for $\F_\varphi$. Since we have
\begin{eqnarray}
(H_{\bm{f}}^{\bm{\alpha}})^n x
&=& \sum_{k=0}^\infty \alpha_k^n \ip{x}{f_k}f_k , \;\;\; x\in D((H_{\bm{f}}^{\bm{\alpha}})^n) , \nonumber \\
(A_{\bm{f}}^{\bm{\alpha}})^n x
&=& \sum_{k=0}^\infty \alpha_{k+1} \alpha_{k+2} \cdots \alpha_{k+n} \ip{x}{f_{k+1}}f_k , \;\;\; x\in D((A_{\bm{f}}^{\bm{\alpha}})^n), \nonumber \\
(B_{\bm{f}}^{\bm{\alpha}})^n x
&=& \sum_{k=0}^\infty \alpha_{k+1}\alpha_{k+2} \cdots \alpha_{k+n} \ip{x}{f_k}f_{k+1}, \;\;\; x\in D((B_{\bm{f}}^{\bm{\alpha}})^n), \nonumber
\end{eqnarray}
it follows that
\begin{eqnarray}
x\in D((H_{\bm{f}}^{\bm{\alpha}})^n)
&{\rm iff}& \sum_{k=0}^\infty \alpha_k^{2n} |\ip{x}{f_k}|^2 <\infty , \nonumber \\
x\in D((B_{\bm{f}}^{\bm{\alpha}})^n)
&{\rm iff}& \sum_{k=0}^\infty (\alpha_{k+1} \cdots \alpha_{k+n})^{2} |\ip{x}{f_{k+1}}|^2 < \infty , \nonumber \\
x\in D((B_{\bm{f}}^{\bm{\alpha}})^n)
&{\rm iff}& \sum_{k=0}^\infty (\alpha_{k+1} \cdots \alpha_{k+n})^{2} |\ip{x}{f_k}|^2 < \infty . \nonumber
\end{eqnarray}
By (4.1), we have
\begin{eqnarray}
\sum_{k=0}^\infty \alpha_{k+1}^{2n} |\ip{x}{f_{k+1}}|^2
&\leq& \sum_{k=0}^\infty (\alpha_{k+1} \cdots \alpha_{k+n})^2 |\ip{x}{f_{k+1}}|^2 \nonumber \\
&\leq& \sum_{k=0}^\infty (\alpha_{k}+(n-1)r)^{2n} |\ip{x}{f_k}|^2 , \nonumber
\end{eqnarray}
and
\begin{eqnarray}
\sum_{k=0}^\infty \alpha_k^{2n} |\ip{x}{f_k}|^2
&\leq& \sum_{k=0}^\infty (\alpha_{k+1} \cdots \alpha_{k+n})^2 |\ip{x}{f_k}|^2 \nonumber \\
&\leq& \sum_{k=0}^\infty (\alpha_k + nr)^{2n}|\ip{x}{f_k}|^2. \nonumber
\end{eqnarray}
Hence it follows that $ x\in D((H_{\bm{f}}^{\bm{\alpha}})^n)\; {\rm iff}\; x\in D((A_{\bm{f}}^{\bm{\alpha}})^n) \; {\rm iff} \; x\in D((B_{\bm{f}}^{\bm{\alpha}})^n)$, which implies that $D^\infty (H_{\bm{f}}^{\bm{\alpha}})=D^\infty (A_{\bm{f}}^{\bm{\alpha}}) =D^\infty (B_{\bm{f}}^{\bm{\alpha}})$. Furthermore, it is clear that $H_\varphi^0$, $A_\varphi^0$, $B_\varphi^0 \in \Lc (T_\varphi^0 D^\infty (H_{\bm{f}}^{\bm{\alpha}}))$. This completes the proof.
\par
(2) This is shown similarly to (1).

\end{proof}

\section*{Conclusions}

This paper continues our (joint, and separate) analysis of biorthogonal sets of vectors of different nature, and their interest in quantum mechanics. In particular, we have shown that the extension of the notion of $\D$-quasi basis can be technically useful and may be of some interest in applications. However, more should be done, mainly on this aspect, and we plan to focus more on physics in a future paper.

\section*{Acknowledgements}
This work was partially supported by the University of Palermo, by the Gruppo Nazionale per la Fisica Matematica (GNFM) and by the Gruppo Nazionale per l'Analisi Matematica, la
Probabilit\`{a} e le loro Applicazioni (GNAMPA) of the Istituto
Nazionale di Alta Matematica (INdAM).


\begin{thebibliography}{99}

\bibitem{bag1} Bagarello, F.: {More mathematics on pseudo-bosons}.  J. Math. Phys., {\bf 54}, 063512 (2013)



\bibitem{bb2017} Bagarello, F., Bellomonte,G.: {Hamiltonians defined by biorthogonal sets}.  J. Phys. A, {\bf 50}, N. 14, 145203 (2017)

\bibitem{bit2013} Bagarello, F., Inoue, A., Trapani, C.: {
Non-self-adjoint hamiltonians defined by Riesz bases}.,   J. Math. Phys., {\bf 55}, 033501, (2014)

\bibitem{bit_jmp2018}Bagarello, F., Inoue, I., Trapani, C.:{ Biorthogonal vectors, sesquilinear forms, and some physical operators}. J. Math. Phys., {\bf 59}, 033506, (2018)
\bibitem{bagrus} Bagarello, F., Russo, F. G.: {A description of pseudo-bosons in terms of  nilpotent Lie algebras}.   Journ. Geom. and Phys, {\bf 125}, 1-11, (2018)
\bibitem{hiro1} Inoue, H.: {General theory of regular biorthogonal pairs and its physical operators}. J. Math. Phys., {\bf 57}, 083511 (2016)

\bibitem{hiro2} Inoue, H.: {Semi-regular biorthogonal pairs and generalized Riesz bases}. J. Math. Phys., {\bf 57}, 113502 (2016)
\bibitem{atsushi} Inoue, H., Takakura, M.: {Regular biorthogonal pairs and pseudo-bosonic operators}, J. Math. Phys., {\bf 57}, 083503 (2016)

\bibitem{hiro_taka} Inoue, H., Takakura, M.: { Non-self-adjoint hamiltonians defined by generalized
Riesz bases}. J. Math. Phys., {\bf 57}, 083505 (2016)



\bibitem{schm} Schm{\"u}dgen, K.: {Unbounded Operator Algebras and Representation Theory}. Birkh{\"a}user-Verlag, Basel (1990)






\bibitem{ct_heisen} Trapani C.: {Remarks on infinite-dimensional representations of the Heisenberg Algebra} in {\em Lie Groups, Differential Equations, and Geometry}. G. Falcone (Ed), Springer 2017.




\end{thebibliography}
\end{document}